\documentclass[11pt,a4paper]{article}
\usepackage[margin=1in]{geometry}
\usepackage[T1]{fontenc}
\usepackage[utf8]{inputenc}
\usepackage[english]{babel}
\usepackage{amsmath,amssymb,amsfonts,amsthm,bm,mathtools}
\usepackage{graphicx}
\usepackage{xcolor}
\usepackage{booktabs}
\usepackage{caption}
\usepackage{natbib}
\usepackage{hyperref}
\usepackage{tikz}
\usetikzlibrary{arrows.meta,positioning,calc,decorations.pathreplacing}
\usepackage{setspace}
\usepackage{alltt}
\usetikzlibrary{arrows.meta, positioning}
\allowdisplaybreaks 
\tikzset{
  state/.style={rectangle, draw, minimum size=20mm}
}
\newcommand{\nothere}[1]{}

\newtheorem{thm}{Theorem}[section]

\newtheorem{lem}{Lemma}[section]

\title{Debiased inference for stochastic treatment interventions with survival outcomes}

\author{Torben Martinussen, Mark Bech Knudsen,
and Helene Rytgaard\\
Section of Biostatistics, University of Copenhagen, Copenhagen, Denmark
}

\begin{document}

\maketitle


\vspace{1cm}

\centerline{\sc Summary}
Estimating the causal effect of a time-dependent treatment on time to death is challenging. In this paper, we formulate the problem using the illness–death model and focus on a stochastic intervention that modifies the hazard governing the transition from no treatment to treatment initiation. Such an intervention can only be implemented at the level of the observed data, whereas the causally valid intervention is defined at the level of the true data-generating process. We provide conditions under which the practically feasible intervention corresponds to the desired causal intervention in the specific setting.
We first consider an intervention in which treatment is initiated at a fixed time point, 
which may subsequently be varied across the relevant time span.
However, the resulting estimand is not pathwise differentiable, preventing the development of assumption-lean inference. To address this, we instead consider a smoothed intervention that assigns treatment within a time window around the target time point, thereby yielding a parameter amenable to semiparametric analysis.
We derive the corresponding efficient influence function and propose a debiased one-step estimator with desirable robustness properties. We investigate its finite-sample performance in a simulation study and apply the method to the classical Stanford Heart Transplant data, as well as to data on treatment delay among couples with unexplained subfertility seeking intrauterine insemination.


\section{Introduction}
\label{intro}
It is well known that judging the effect of a time-dependent exposure on time-to-death is  subtle. A classical example is the Stanford Heart Transplant study where patients with end-stage heart disease were placed on a heart
transplant waiting list. For each patient, follow-up began at the time of entry into the waiting list 
and continued until death or censoring. During follow-up, some patients received a heart
transplant, while others died before transplantation or were censored. Using  heart transplant as information given at time zero creates immortal-time bias, and a standard method to deal with this problem has been to use the Cox-model  with a time-dependent exposure (heart transplant); see, e.g., \cite{kalbfleisch2002statistical}. However, the causal interpretation of results from such analyses is subtle, see \cite{hernan2010hazards,martinussen2020subtleties}. This  is also closely related to the proposal by \cite{snapinn2005illustrating} on an extension of the Kaplan-Meier estimator with time-varying covariate(s), which has also been criticized as lacking proper causal interpretation except for some unrealistic cases, see \cite{sjolander2020cautionary,knudsen2025limitations}. 

There exists causal inference methods related to this problem.
For example, \cite{zhang2012estimating} study the estimation of the effect of a randomized treatment in the presence of a secondary  treatment, with the latter being time-dependent. They use the marginal structural Cox model while exploiting the randomization to improve efficiency.
\cite{hernan2018estimate}  deals with  estimation of  the effect of treatment duration (time-varying) on a survival
outcome using a cloning technique and then performing inverse
probability weighting to account for an imposed dependent censoring mechanism.
\cite{prosepe2025causal} formulates the time-dependent treatment issue using the illness-death model with "illness" corresponding to the state when the treatment takes place. They use working Cox-models to estimate needed transition hazards in order to estimate their proposed estimand.
All of these methods rely heavily on correct specification of certain models used in their proposed estimation procedures, and will give  biased results if these models are misspecified which are likely to happen in reality. 
\cite{niessl2020multistate} and \cite{erdmann2023connection} consider different censoring approaches to handle the situation where a treatment is discontinued and can be restarted later in time. To censor individuals is inefficient use of data and can also lead to bias if the induced censoring is affected by unobserved variables that may also be predictive of death.

Similar to \cite{prosepe2025causal}, we also make use  of the illness–death model as the canonical framework
to evaluate the causal effect of time-dependent treatment.
Our contribution differs conceptually from standard marginal structural models and g-computation approaches in that the intervention is formulated directly through modifications of the transition intensity governing treatment initiation within the illness–death framework. This allows us to distinguish explicitly between interventions defined at the level of observed intensities and interventions defined at the level of the underlying data-generating process. The latter distinction is particularly important in the presence of latent variables affecting both treatment initiation and competing or downstream transitions.
The intervention of interest modifies the treatment initiation hazard, thereby defining a stochastic intervention on the data-generating process. We show that, in the presence of unmeasured variables affecting both the treatment decision and competing or downstream hazards, interventions defined at the level of observed intensities may fail to correspond to causally valid interventions. Focusing on treatment at a fixed time point, we propose  the interventional risk as  the target estimand.

A key technical difficulty is that interventions fixing treatment initiation at an exact time point lead to estimands that are not pathwise differentiable.
To address this, we introduce a smoothed intervention that assigns treatment within a time window around the target time point, yielding a well-defined parameter amenable to semiparametric analysis. We derive the corresponding efficient influence function for this parameter under general conditions, including right censoring. 
Building on this semiparametric representation, we propose a debiased one-step estimator that combines plug-in estimation with influence-function-based correction. The estimator accommodates flexible, data-adaptive estimation of nuisance parameters, including machine learning approaches, and allows for cross-fitting. We establish asymptotic normality and show that the estimator enjoys rate double robustness with respect to the nuisance components. Extensions to settings with baseline covariates are provided, along with simulation studies 
that demonstrate the theoretical properties of the proposed method.
We further apply the method to the Stanford Heart Transplant study as well as to data concerning treatment
delay on a cohort of couples with unexplained subfertility who seek intrauterine insemination \citep{prosepe2025causal}.

\section{A causally valid stochastic intervention in the  illness-death setting}

\subsection{Preamble}

We define the setting in the case where  there is no censoring and use the analogy of the illness-death model with three states, see Figure 1,
so that all individuals start in state 0 (all are untreated at time zero), and treatment can then be initiated at some point in time thereafter corresponding to a transition to state 1. Because the endpoint that we consider in this paper is time-to-death, a transition to this state (state 2) may occur before treatment is initiated. This can also be described using 
 $Z=\{\delta=I(U\leq T),T,U\wedge T\}$,  where $U$ denotes the time to a 01-transition, which is only observed if $\delta=1$. The time to death is denoted by $T.$
 We thus have the following counting processes and at-risk processes: $N_{01}(t)=I(U\wedge T\leq t, \delta=1)$, $N_{02}(t)=I(U\wedge T\leq t, \delta=0)$, $Y_{0}(t)=I(t\leq U\wedge T)$, and $N_0(t)=N_{01}(t)+N_{02}(t)=I(U\wedge T\leq t)$. Also, $N_{12}(t)=I(U<T\leq t)$ and $Y_1(t)=I(U<t\leq T)$.  Observation takes place in the time period $[0,\tau]$ with $\tau<\infty$.

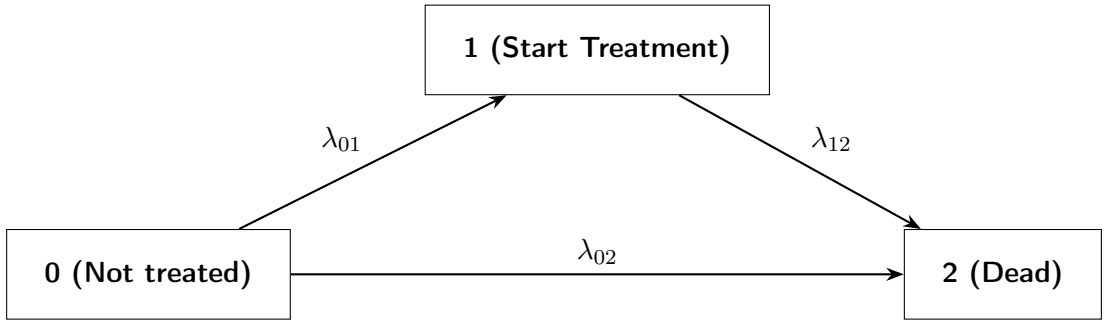
\begin{figure}[h!]
\label{Fig: Illness-death}
    \centering
    \begin{tikzpicture}[
      node distance=2.5cm,
      state/.style={rectangle, draw, align=center, inner xsep = 0.5cm, inner ysep = 0.4cm, font=\sffamily\bfseries},
      arrow/.style={-{Stealth}, thick}
    ]
    
    \node[state] (0) {0 (Not treated)};
    \node[state, above right = of 0] (1) {1 (Start Treatment)};
    \node[state, below right = of 1] (2) {2 (Dead)};
    
    \draw[arrow] (0) -- node [above left] {$\lambda_{01}$} (1);
    \draw[arrow] (1) -- node [above right] {$\lambda_{12}$} (2);
    \draw[arrow] (0) -- node [above] {$\lambda_{02}$} (2);
    
    \end{tikzpicture}
    \caption{Illness-death model.}
    \label{fig:extended_illness_death}
\end{figure}

\vspace{1cm}

 \subsection{Defining the intervention}

We use ${\cal F}_t$ to denote the full filtration, which refers to the filtration containing the information associated with the data generating process (DGP), and we let ${\cal O}_t\subseteq {\cal F}_t$ denote  the observed filtration, where we still consider the situation without censoring. 
We use the notation 
$$
E\{dN_{01}(t)|{\cal F}_{t-}\}=\lambda_{01}(t) dt=Y_{0}(t)\alpha_{01}(t)dt,
$$
where $\lambda_{01}(t)$ is the ${\cal F}_t$-intensity process associated with $N_{01}(t)$, $Y_{0}(t)$ and $\alpha_{01}(t)$ are the corresponding at-risk indicator and transition hazard function, respectively.
We use similar notation for the other two counting processes $N_{02}(t)$ and
$N_{12}(t)$. Note that the intensity process $\lambda_{12}(t|u)=Y_1(t)\alpha_{12}(t|u)$ may depend on when the transition to state 1 takes place, that is, time point $u$. These intensity processes govern the DGP. Note that there may be information, summarized in $V$, contained in  ${\cal F}_{t}$, which affects some (or all) of the transition intensities but this $V$ may not be available in ${\cal O}_t$. The ${\cal O}_t$-transition hazards are denoted by $\alpha_{jk}^{o}$ with $jk\in\{01,02,12\}$.

The intervention we would like to consider, had we access to the DGP,  is the one that sets $\alpha_{01}(t)$ to some $\alpha^{*}_{01}(t)$. It is important to stress that the causally valid intervention takes place at the DGP level. However, since we may not have access to the information given by ${\cal F}_t$, but only the information given by ${\cal O}_t$,  we can only perform the intervention at the ${\cal O}_t$ local characteristics (the intensities), that is setting $\alpha^{o}_{01}(t)$ to $\tilde \alpha^{o}_{01}(t)$.
In related ongoing work (Knudsen, Rytgaard and Martinussen, manuscript in preparation), we show that for such interventions to correspond to the causally valid intervention, no unobserved variable may be a common cause of the treatment initiation hazard and competing or downstream hazards whose structure is kept fixed.
We illustrate this with the following example.
\bigskip

{\bf Example.}
Consider the following situation
\begin{align}
\label{DGP-intens}
    \lambda_{02}(t)=Y_{0}(t) V\alpha_{02}(t), \quad
  \lambda_{01}(t)=Y_{0}(t) V\alpha_{01}(t), \quad
   \lambda_{12}(t)=Y_{1}(t) \alpha_{12}(t),
\end{align}
where $V$ is an unmeasured frailty, which we take, for ease of calculations, as  $V\sim \Gamma(1,1)$, although this is not important for the general argument. 
The intervention we would like to evaluate is setting $\alpha_{01}^{F}(t,V)\equiv V \alpha_{01}(t)$ to something specific such as "never treat", i.e., $\alpha_{01}^{F}(t,V)=0$.
In practice, we can only evaluate the intervention $\alpha_{01}^{o}(t)=0$. Our general result implies that this feasible intervention does not correspond to the one we would like to carry out had we known the DGP, as, in this case, it involves the unmeasured variable $V$.
Specifically, we see that the  unmeasured variable
 is a parent of the decision hazard as well as a parent of the competing hazard $\alpha_{02}^F(t,V)\equiv V\alpha_{02}(t)$. We illustrate this with concrete calculations to see explicitly what goes wrong.
We get in this specific setting,
\begin{align*}
E\{dN_{02}(t)|{\cal O}_{t-}\}&=Y_{0}(t)\frac{\alpha_{02}(t)dt }{1+  A_{02}(t)+  A_{01}(t)},\quad
E\{dN_{01}(t)|{\cal O}_{t-}\}=Y_{0}(t)\frac{\alpha_{01}(t)dt }{1+  A_{02}(t)+  A_{01}(t)}\\
E\{dN_{12}(t)|{\cal O}_{t-}\}&=Y_{1}(t)\alpha_{12}(t),
\end{align*}
where $A_{0j}(t)=\int_0^t\alpha_{0j}(s)\, ds$, $j=1,2$.
\nothere{The  intervention "never treat", i.e.,  setting $  \alpha_{01}(t)=0$ and 
$\alpha^{o}_{01}(t)=\tilde \alpha^{o}_{01}(t)=0$ results in
$\alpha_{01}^{*o}(t)=\tilde \alpha_{01}^{o}(t)=0$, but 
$$
\alpha_{02}^{*o}(t)=\frac{\alpha_{02}(t)}{1+ A_{02}(t)},\quad \tilde \alpha_{02}^{o}(t)=\frac{\alpha_{02}(t)}{1+  A_{02}(t)+ A_{01}(t)},
$$
respectively.}
The intervention ''never treat'' is correctly specified by setting $\alpha^F_{01}(t, U) = \alpha^*_{01} (t)= 0$, resulting in observed hazards
$$
\alpha^{*o}_{01}(t) = 0, \quad \alpha_{02}^{*o}(t)=\frac{\alpha_{02}(t)}{1+ A_{02}(t)}.
$$
Specifying the intervention directly on the observed level amounts to setting $\alpha^o_{01}(t) = \tilde{\alpha}^o_{01}(t) = 0$, resulting in observed hazards
$$
\tilde{\alpha}^o_{01}(t) = 0, \quad \tilde \alpha_{02}^{o}(t)=\frac{\alpha_{02}(t)}{1+  A_{02}(t)+ A_{01}(t)}.
$$
Both formulations agree that the observed $0 \to 1$ hazard is 0, but they disagree on the observed $0 \to 2$ hazard.
We use $\tilde P^{o}$ to denote probability measure under the intervention at the   ${\cal O}_t$-level while $P^*$ is used to denote the probability measure under the intervention at the ${\cal F}_t$-level.
Thus,
$$
\tilde P^{o}(T\leq t)\neq  P^{*o}(T\leq t)
$$
unless $\alpha_{01}(t)=0$ from the outset. Hence, we are  misled by $\tilde P^{o}(T\leq t)$ because of the unmeasured confounder $V$. 
Had the DGP instead been given by \eqref{DGP-intens}, but with $\lambda_{02}(t)=Y_{0}(t) \alpha_{02}(t)$ then 
$$
\alpha_{02}^{*o}(t)=\tilde \alpha_{02}^{o}(t)=\alpha_{02}(t)
$$
and further $\tilde P^{o}(T\leq t)=P^{*o}(T\leq t)$.
\hfill $\square$

In the rest of the paper, we assume that intervening on $\alpha_{01}^{o}$ is causally valid.

\section{Treat at time point $u^*$}
\label{Sect:Treat at time point}
We have in the illness-death setting, without any intervention, that
\begin{align}
\label{Prob-die}
    P^{o}(T\leq t)=&\int_0^t\exp \left\{-A^{o}_{01}(u)-A^{o}_{02}(u)\right\}\alpha^{o}_{02}(u)du\notag\\
    &+\int_0^t\exp \left\{-A^{o}_{01}(u)-A^{o}_{02}(u)\right\}\alpha^{o}_{01}(u)\left\{\int_u^t \exp \left(-\int_u^s \alpha^{o}_{12}(v|u) \mathrm{d} v\right)\alpha^{o}_{12}(s|u) 
    ds\right\}du,
\end{align}
where
$A^{o}_{0j}(t)=\int_0^t\alpha^{o}_{0j}(s)\, ds$, $j=1,2$.
In the following, we consider the intervention corresponding  to treat at time point $u^*$ meaning that we set $\alpha^{o}_{01}(t)$ to zero for $t<u^*$ and then let it spike towards infinity at $u^*$. 
This results in 
$$
\tilde P^{o}\left (T\leq t\right )
= \left\{ \begin{array}{ll}
\int_0^t\exp \left\{-A^{o}_{02}(u)\right\}\alpha^{o}_{02}(u)du & t\leq u^*\\
1-\exp \left\{-A^{o}_{02}
(u^*)\right\}\left\{1-\int_{u^*}^t \exp \left(-\int_{u^*}^s \alpha^{o}_{12}(v|u^*)dv \right )\alpha^{o}_{12}(s|u^*)ds  \right\}& t>u^*
  \end{array} \right.
$$
which is currently the estimand of interest, assuming that there are no unobserved variables that are  the parent of both the decision hazard and competing or downstream hazards as explained in the latter section. We wish to perform assumption lean inference, which requires calculation of the associated efficient influence function (EIF). 
Unfortunately, the above estimand is not pathwise differentiable because of the conditioning in $\alpha^{o}_{12}$, and therefore the corresponding 
 EIF does not exist.
The underlying issue is that the intervention fixes treatment initiation at an exact time point in continuous time. Consequently, the estimand depends on local behavior of the conditional post-treatment hazard around a single point $u^*$, which cannot be regularized through standard first-order perturbations of the observed data distribution.

This lack of pathwise differentiability reflects the instability of interventions assigning treatment at an exact time point in continuous time. Instead, we propose a  smoothed intervention serving two purposes: it yields a practically meaningful treatment regime allowing treatment initiation within a clinically realistic time window, and it regularizes the estimand sufficiently to permit semiparametric inference.
The intervention initiates
treatment in a time window around $u^*$ as follows.
Define $g(u)=\alpha_{01}^{o}(u)e^{-A_{01}^{o}(u)}$ and
$$
\tilde g(u)=g(u)\frac{I(u^*-b<u<u^*+b)}{c^*},
$$
where 
$$
c^*=\int_{u^*-b}^{u^*+b}g(u)\, du,
$$
and $b$ is a bandwidth parameter. At the boundaries (0 and $\tau$), we need to replace $u^*-b$ with $0\vee (u^*-b)$ and $u^*+b$ with $\tau\wedge (u^*+b)$. 
To simplify notation, we use $u_-^*=0\vee (u^*-b)$ and $u_+^*=\tau\wedge (u^*+b)$ throughout the remainder of the paper.
The bandwidth parameter $b$ therefore controls the width of the treatment initiation window and should be viewed as part of the target intervention itself rather than merely a tuning parameter.

Further, let 
$$
\tilde A^{o}_{01}(s)
= \left\{ \begin{array}{lll}
0& s\leq u_{-}^*\\
-\log{\left\{\int_s^{u_+^*}\tilde g(v)dv\right\}} & u_{-}^*<s<u_+^*\\
 \infty& s\geq u_+^*
  \end{array} \right.
$$
or
$$
\tilde \alpha^{o}_{01}(s)
= \left\{ \begin{array}{lll}
0& s\leq u_{-}^*\\
\frac{\tilde g(s)}{\int_s^{u_+^*}\tilde g(v)dv} & u_{-}^*<s<u_+^*\\
 \infty & s\geq u_+^*
  \end{array} \right.
$$
from which  it is seen that $\tilde \alpha^{o}_{01}(s)$ converges to $\infty $ when $s$ goes to $u^*+b$ from the left.
Furthermore, if we let $b$ increase towards infinity, we see that
$\tilde \alpha^{o}_{01}(s)$ converges to $g(s)/e^{-A_{01}^{o}(s)}=\alpha^{o}_{01}(s)$.

The estimand of interest is $\psi_t=\tilde P^{o}\left (T\leq t\right )$, i.e., in \eqref{Prob-die} replace  $ \alpha^{o}_{01}(s)$ with $\tilde \alpha^{o}_{01}(s)$. This gives
$$
\psi_t=\psi_{t}^{-}I(t<u_{-}^*)+\psi_{t}^{+}I(t\geq u_{-}^*),
$$
where
$$
\psi_{t}^{-}=\int_0^{t}e^{-A^{o}_{02}(s)}dA^{o}_{02}(s)
$$
and $\psi_{t}^{+}=\psi_{u_-^*}^{-}+\mu_t+\eta_t$ with 
\begin{align}
\label{mut}
   \mu_t=
    \int\frac{J^*_t(u)}{c^*}H_{01}^{*}(u)\,  e^{- A^{o}_{02}(u)}d A^{o}_{02}(u),\;
    \eta_t=\int\frac{J^*_t(u)}{c^*} \{1-S_1(t|u)\}S_0(u)\alpha_{01}^{o}(u)\, du,
\end{align}
\begin{align*}
    S_0(u)=e^{- A^{o}_{01}(u)-A^{o}_{02}(u)},\quad
S_1(t|u)=\exp \left\{-\int_u^t \alpha^{o}_{12}(v|u) \mathrm{d} v\right\},
\end{align*} 
$H_{01}^{*}(u)=\int_u^{u^*+b}g(v)dv$ and $J^*_t(u)=I\{u_-^*<u<t\wedge u_+^*\}$. We also use the notation 
$S_{0j}(u)=e^{- A^{o}_{0j}(u)}$, $j=1,2$, so $S_0(u)=S_{01}(u)S_{02}(u)$.

 \subsection{Efficient estimation of interventional risk}
\bigskip

In the case where  there is no censoring, 
we have the martingales w.r.t ${\cal O}_{t}$:   $M_{0j}(t)=N_{0j}(t)-\int_0^tY_0(s)\alpha^{o}_{0j}(s)\, ds$, $j=1,2$, $M_{0}(t)=N_{0}(t)-\int_0^tY_0(s)\{\alpha^{o}_{01}+\alpha^{o}_{02}\}(s)\, ds$, $M_{12}(t)=N_{12}(t)-\int_0^tY_1(s)\alpha^{o}_{12}(s|U)\, ds$. 
In reality there will be censoring, and we denote the 
observed data  
$O=\{T\wedge C,U\wedge T \wedge C, \delta_1=I(U\leq T\wedge C),\delta_2=I( T\leq C)\}$. Here, we have used $C$ to denote the censoring variable that has survival function denoted by $K$ and hazard function $\lambda_C$. Note that $K(r|u)$ may depend on the timing $u$ at which treatment is initiated with $r>u$.
We assume that the censoring results in coarsening at random (CAR) as defined in \cite{tsiatis2006semiparametric}.
We refer to the observed data martingales by $M_{01}^{o}(t)$, and similarly for the other martingales. We also modify the counting processes as well as the at-risk processes, and refer to them similarly, i.e., $N_{01}^{o}(t)=I(U\wedge T \wedge C \leq t,\delta_1=1)$.

The EIF's given in Lemma \ref{basic_eif} in the Appendix are the required building blocks to calculate the desired EIF, $D^*_{\psi_t}(O,P) $.
The following result is obtained.
\begin{thm}
   \label{thm:eif} 
The observed data efficient influence function $D^*_{\psi_t}(O,P)$ corresponding to $\psi_t$ 
is $I(t<u_-^*)D^*_{\psi^-_t}(O,P)+I(t\geq u_-^*)D^*_{\psi^+_t}(O,P)$,
where 
$$
D^*_{\psi^-_t}=e^{-A_{02}(t)}\int_0^t\frac{dM_{02}^{o}(u)}{S_0(u)K(u)}, \quad t<u^*-b,
$$
and, for $t\geq u_-^*$, $D^*_{\psi^+_t}(O,P)=D^*_{\psi^-_{u_-^*}}(O,P)+D^*_{\mu_t}(O,P)+D^*_{\eta_t}(O,P)$ with 
\begin{align*}
D^*_{\eta_t}(O,P)=&\int J^*_t(u)h_{\eta}^{(1)}(u,t)\frac{dM_{01}^{o}(u)}{S_0(u)K(u)}
 +\frac{\delta_1J_{t}^*(U)}{c^*K(U)}S_1(t|U)\int _U^t\frac{dM_{12}^{o}(r|U)}{S_1(r|U)K(r|U)}\\
&
-\int_{u_-^*}^{t\wedge u_+^*} \int_0^{u}\frac{dM_0^{o}(v)}{S_0(v)K(v)}
 \{1-S_1(t|u)\}S_0(u)\alpha_{01}^{o}(u)\, du
-\eta_t D^*_{c^*}(O,P)/c^*,\\
D^*_{\mu_{t}}(O,P)=
&\int_{0}^{t\wedge u_+^*} h_{\mu}^{(1)}(u,t) \frac{dM_{01}^{o}(u)}{S_0(u)K(u)} 
+\int_{0}^{t\wedge u_+^*} h_{\mu}^{(2)}(u,t)\frac{dM_{02}^{o}(u)}{S_0(u)K(u)}
-\mu_{t} D^*_{c^*}(O,P)/c^*.
\end{align*}
where $M_0^{o}=M_{01}^{o}+M_{02}^{o}$ and $h_{\eta}^{(1)}(u,t)$,
$h_{\mu}^{(j)}(u,t)$, $j=1,2$, are defined in the proof. Furthermore,
$$
D^*_{c^*}(O,P)=\int\frac{v_{01}(u)dM_{01}^{o}(u)}{S_0(u)K(u)},
$$
where $v_{01}(u)=I(u<u_+^*)S_{01}(u_+^*)-I(u<u_-^*)S_{01}(u_-^*)$.
\end{thm}

The proof is given in  the Appendix. 
Since we have a model-free expression for $\psi_t(P)$, it is tempting to substitute 
all unknown distributional quantities by
data-adaptive estimators (e.g., machine learning based estimators)
to obtain a  plug-in estimator $\hat\psi_t^{{pi}}=\psi_t(P_n)$, where we by $P_n$ mean that unknown quantities in $\psi_t$ are replaced with estimated counterparts. For the current estimand, we can use Nelson-Aalen estimators for $\{A_{01}^{0},A_{02}^{0}\}$, 
so that the main remaining component is the estimation of $S_{1}(r|u)$, the conditional survival function after treatment initiation. While, in principle, this component can be estimated using flexible data-adaptive methods, it is well known that the resulting plug-in estimator may exhibit poor finite-sample performance, as the bias–variance trade-offs involved in estimating $S_{1}(r|u)$ are not aligned with the target parameter $\psi_t$.
To address this, we consider a one-step (debiased) estimator \citep{kennedy2022semiparametric}, given by
$$
\hat\psi_t^{os}=\hat\psi_t^{pi}
+\mathbb{P}_n\left\{D^*_{\psi_t}(O,P_n)\right \},
$$
where  $\mathbb{P}_n\{v(O)\}=n^{-1}\sum_i v(O_i)$ is the empirical measure.
This estimator corrects, to first order, the bias of the plug-in estimator via the efficient influence function.
In the implementation, we assume that all nuisance parameters are replaced by consistent estimators.
This is straightforward for 
 $A_{01}^{0}$ $A_{02}^{0}$,
 which may be estimated using Nelson–Aalen estimators.
 In principle,
 flexible data-adaptive procedures can be used for the estimation of $S_{1}(r|u)$ and the censoring distribution $K(r|u)$.

Theorem \ref{thm:remainder}, whose proof is provided in the Appendix, states the limit distribution of the resulting estimator. Since we will use super/sub-script $n$ to denote estimators, we will from now on drop the superscript $o$ from the observed transition hazard function as well as from their cumulative counterparts.

\begin{thm}\label{thm:remainder}
Consider $n$ i.i.d. replicates of $O=\{T\wedge C,U\wedge T \wedge C, \delta_1=I(U\leq T\wedge C),\delta_2=I( T\leq C)\}$.
Then $\hat\psi_t^{os}$ is asymptotically normally distributed with mean $\psi_t$
and a variance that can be consistently estimated as the sample variance of $D^*_{\psi_t}(O,P_n)$,
provided that the nuisance parameter estimators 
$A^n_{12}(r|u)$  and $K_n(r|u)$ are trained on a separate part of the data, and that   
\begin{align}
\label{robprop}
    \int J_t^*(u)&\left\{\int_u^t\frac{\{K_n(r|u)-K(r|u)\}}{K_n(r|u)}\frac{S_1(r|u)}{S^n_1(r|u)}d\{ A_{12}^n(r|u)- A_{12}(r|u) \}\right\}\nonumber\\
&\times S_1^n(t|u)S_0^n(u)dA^n_{01}(u)=o_p(n^{-1/2}).
\end{align}

\end{thm}

In particular, condition \eqref{robprop} highlights a form of rate double robustness between the estimation of the post-treatment survival function $S_1(t|u)$ and the censoring distribution. Thus, sufficiently fast convergence of one component may compensate for slower convergence of the other. In principle, this allows for flexible, data-adaptive estimation of the nuisance parameters, although in practice the estimation of $S_1(t|u)$ and the censoring distribution $K(r|u)$  may be challenging due to its dependence on the treatment initiation time $u$ and that the estimation procedure needs to be able to handle 
delayed-entry;  to our knowledge such fully data-adaptive procedures are not yet  available.


\section{Setting with baseline covariates}
We now extend the method described in the previous section to the setting with baseline covariates that we denote by $L$. All intensity functions are specified given $L$ and the CAR assumption is now conditional on the additional baseline information.
Define $g(u|l)=\alpha_{01}^{o}(u|l)e^{-A_{01}^{o}(u|l)}$ and
$$
\tilde g(u|l)=g(u|l)\frac{I(u_-^*<u<u_+^*)}{c^*(l)},\quad c^*(l)=\int_{u_-^*}^{u_+^*}g(u|l)\, du.
$$
Further, let 
$$
\tilde \alpha^{o}_{01}(s|l)
= \left\{ \begin{array}{lll}
0& s\leq u_-^*\\
\frac{\tilde g(s|l)}{\int_s^{u_+^*}\tilde g(v|l)dv} & u_-^*<s<u_+^*\\
 \infty & s\geq u_+^*
  \end{array} \right.
$$
We consider the intervention where  $\alpha_{01}^{o}(s|l)$ is replaced  with $\tilde \alpha_{01}^{o}(s|l)$, and the resulting
 estimand  is 
$\psi_t=\tilde P^{o}(T\leq t)=\int \tilde P^{o}(T\leq t|l)f_L(l)\, dl $, where
$\tilde P^{o}(T\leq t|l)$ refers to the conditional probability measure under the considered intervention and $f_L(l)$ is the density function corresponding to the marginal distribution of $L$.

Define $\psi_{t}^{-}(l)$, $\mu_t(l)$, $\eta_t(l)$, $\psi_{t}^{+}(l)$ and $\psi_{t}(l)$ as in Section \ref{Sect:Treat at time point} replacing all transition hazards with the conditional ones as defined above, and let $\psi_{t}^{-}=\int \psi_{t}^{-}(l) f_L(l) dl$, and similarly for $\mu_t$, $\eta_t$, $\psi_{t}^{+}$ and $\psi_{t}$.
It follows that 
$$
D^*_{\psi_{t}^{-}}(O,P)=\psi_{t}^{-}(L)-\psi_{t}^{-}+e^{- A^{o}_{02}(t|L)}\int_0^{t}\frac{dM_{02}^{o}(s|L)}{S_0(s|L)K(s|L)}.
$$
For $t<u^*-b$, $\psi_t=\psi_{t}^{-}$, and the corresponding one-step estimator is 
$$
\hat\psi_t^{os}=\mathbb{P}_n\left\{ \psi_{tn}^{-}(L)+e^{- A^{o}_{02,n}(t|L)}\int_0^{t}\frac{dM_{02,n}^{o}(s|L)}{S_0^n(s|L)K_n(s|L)}\right \}
$$
\nothere{
Let  
$$
\mu_{2t}(L)=\frac{1}{c^*(L)}\int J^*_t(u)\int_u^{u^*+b}g(v|L)dv\, e^{- A^{o}_{02}(u|L)}dA^{o}_{02}(u|L)
$$
so that $\mu_{2t}\equiv \int \mu_{2t}(l) f_L(l)dl$ corresponds to the second term on the right-hand side of \eqref{mut}. The term corresponding to the first  term on the right-hand side of \eqref{mut} is handled similarly to what we just derived.

Also, let  
$$
\eta_t(l)=\frac{1}{c^*(l)}\int J^*_t(u)Q_{12}(u,t|l)S_0(u|l)\alpha_{01}(u|l)\, du,
$$
and $\eta_t=E\{\eta_t(L)\}.$}
For the general case, we have the following result.
\begin{thm}
   \label{eif.cov} 
The observed data efficient influence function corresponding to $\psi_t$ is 
\begin{align*}
D^*_{\psi_t}(O,P)=\psi_t(L)-\psi_t+D^*_{\psi_t}(O,P|L),
\end{align*}
where $D^*_{\psi_t}(O,P|L)$ means the expression given in Theorem \ref{thm:eif}, but with all distributional quantities replaced by the corresponding ones where we condition on $L$, so that for instance $dM_0^{o}(u)$ is replaced by $dM_0^{o}(u|L)$ and similarly for all other distributional quantities.
\end{thm}
Thus, in the general case, the one-step estimator is  given by 
$$
\hat\psi_t^{os}=\mathbb{P}_n\left\{ \psi_{tn}(L)+D^*_{\psi_t}(O,P_n|L)\right \}.
$$

\begin{thm}\label{thm:remainder_w_L}
Consider $n$ i.i.d. replicates of
$O=\{T\wedge C,U\wedge T \wedge C,
\delta_1=I(U\leq T\wedge C),\delta_2=I(T\leq C),L\}$.
Then $\hat\psi_t^{os}$ is asymptotically normally distributed with mean $\psi_t$
and variance consistently estimated by the sample variance of
$D^*_{\psi_t}(O,P_n)$, provided that the nuisance estimators are trained on a
separate part of the data and that the second-order remainder is
$o_p(n^{-1/2})$.

More explicitly, the remainder is a finite sum of weighted integrals whose
integrands contain products of nuisance estimation errors of the following types:
\begin{align*}
&\{S_1^n(t|u,L)-S_1(t|u,L)\}\{S_{02}^n(u|L)-S_{02}(u|L)\},\\
&\{S_1^n(t|u,L)-S_1(t|u,L)\}
  \{S_{01}^n(u|L)dA_{01}^n(u|L)-S_{01}(u|L)dA_{01}(u|L)\},\\
&\{S_{02}^n(u|L)-S_{02}(u|L)\}
  \{S_{01}^n(u|L)dA_{01}^n(u|L)-S_{01}(u|L)dA_{01}(u|L)\},\\
&\{S_{02}^n(u|L)-S_{02}(u|L)\}\{dA_{01}^n(u|L)-dA_{01}(u|L)\},\\
&\{S_{01}^n(u|L)-S_{01}(u|L)\}\{dA_{02}^n(u|L)-dA_{02}(u|L)\},\\
&\frac{K_n(r|u,L)-K(r|u,L)}{K_n(r|u,L)}
  d\{A_{12}^n(r|u,L)-A_{12}(r|u,L)\},\\
&\frac{K_n(u|L)-K(u|L)}{K_n(u|L)}
  d\{A_{0j}^n(u|L)-A_{0j}(u|L)\},\quad j=1,2.
\end{align*}
Here all products are integrated against bounded weights over the relevant
time regions. 
Sufficient conditions are that each of these weighted integrals is
$o_p(n^{-1/2})$.
\end{thm}
The display above emphasizes the types of nuisance estimation error products
appearing in the second-order remainder rather than an algebraically exhaustive
decomposition. The key point is that all remainder contributions are of second
order and involve products of estimation errors from distinct nuisance
components.
This directly reveals the robustness structure of the estimator. In particular,
slow convergence of one nuisance component may be compensated by sufficiently
fast convergence of the component with which it is paired.
In principle, this allows for flexible, data-adaptive estimation of the nuisance parameters, including machine learning approaches. In practice, however, some components - most notably the post-treatment transition and the associated conditional survival function $S_1(t|u,L)$ - may be more challenging to estimate accurately due to their dependence on both treatment initiation time and covariates as well as the built-in delayed-entry structure. This motivates the use of a combination of parametric working models and data-adaptive estimators in finite-sample implementations.

{\bf Remark}

\vspace{-0.8cm}

\begin{itemize}
     \item[] We may look for simple interaction effects as follows.
    If we split the covariates into $(W,L)$ with $W$ categorical (binary, say),  we may then consider
 $$
 \psi_t(w)=\int \widetilde P^{o}(T\le t\mid w,L=l)\,f_{L|W}(l|w)\,dl.
 $$
Next step is to compare $ \psi_t(1)$ and $ \psi_t(0)$ under various treatment regimes. With $q(w)=P(W=w)$, we get the plug-in estimator
 $$
 \hat \psi^{pi}_t(w)=\frac{1}{nq(w)}\sum_iI(W_i=w)\tilde P_n(T\le t\mid w,L_i)
 $$
 and the one-step estimator is 
 $$
\hat \psi^{os}_t(w)=\hat \psi^{pi}_t(w)+\frac{1}{nq(w)}\sum_iI(W_i=w)D_{\psi_t}^*(O,P_n\mid w,L_i).
 $$
 As an illustration, we apply this in Section \ref{STH} when re-analyzing the Stanford Heart Transplant data.

 \end{itemize}

 \vspace{-2cm}
 
\section{Numerical studies}

\subsection{Simulation study}
We investigated the performance of the proposed procedure in simulations using a setup
where $L$ was two-dimensional with $L_1$ and $L_2$ independent, $L_1$ standard normal and $L_2$ binary with $P(L_2=1)=0.5$. We took 
$\alpha_{01}(t|L)=0.4\exp{( 0.25L_1 )}$, $\alpha_{02}(t|L)=0.4\exp{( -0.25L_1 )}$ and 
$\alpha_{12}(t|u,L)=0.4\exp{(-0.2u- 0.25L_1 )}$. Censoring was generated using $ \alpha_C(t|L)=0.2\exp{(0.1 L_1 )}$.
All working models included both covariates $L_1$ and $L_2$, although the data-generating process does not depend on $L_2$.
We took $b=0.3$, $u^*=0.8$, $t=2$ and used $n=300, 600.$  
As noted in the previous section, we adopt a pragmatic approach to the estimation of the nuisance components. 
The purpose of the simulations is primarily to assess the robustness properties of the proposed one-step correction under nuisance-model misspecification rather than to demonstrate a fully machine-learning-based implementation.
We first use working Cox models for all transition hazards, reflecting common practice in applied survival analysis and allowing us to isolate the theoretical properties of the proposed estimator. In particular, we consider both correctly specified working models and a scenario  in which the post-treatment transition model is misspecified.
Under correct specification, we compute both the plug-in estimator $\hat\psi_t^{pi}$ and the one-step estimator $\hat\psi_t^{os}$. Results are reported in Table \ref{Tab1}, where both estimators are seen to be unbiased, and the variability of the one-step estimator is well captured by the corresponding estimated standard error.
We then modify the DGP by changing only the post-treatment transition to 
$\alpha_{12}(t|u)=0.4\exp{\{-0.4(u-1.22)^2- 0.15L_1 \}}$, while keeping the working models unchanged. 
The working model for the 12-transition is therefore misspecified. The corresponding results are  reported in Table \ref{Tab1}, showing that the plug-in estimator is now biased, whereas the one-step estimator remains approximately unbiased, in agreement with the theoretical results.

To further illustrate that the proposed framework is compatible with flexible nuisance estimation, we additionally consider data-adaptive estimation of selected nuisance components using machine learning methods. Specifically, we apply machine learning to the estimation of the 01- and 02-transition hazards and the censoring distribution, for which standard survival learning tools \citep{westling2024inference} are readily available, combined with 3-fold cross-fitting. Results are marked with ML in Table \ref{Tab1}. The resulting estimator continues to exhibit the expected bias correction, supporting the theoretical robustness properties of the proposed method. A slight bias is observed when $n=300$, but this diminishes as the sample size increases.

\begin{table}[ht]
\centering
\begin{tabular}{llcccccc}
\toprule
Working model & $n$ & True value & Mean ($\hat\psi_t^{pi}$) 
& Mean ($\hat\psi_t^{os}$) & SD & ESE & Cov95 (\%) \\
\midrule
Correct  
& 300 & 0.521 & 0.519 & 0.521 & 0.067 & 0.067 & 94.1 \\
Correct
& 600 & 0.521 & 0.519 & 0.518 & 0.049 & 0.047 & 94.7 \\
Misspecified
& 300 & 0.533 & 0.501 & 0.530 & 0.067 & 0.068 & 94.4 \\
Misspecified
& 600 & 0.533 & 0.501 & 0.530 & 0.049 & 0.048 & 94.6 \\
Misspecified, ML
& 300 & 0.533 & 0.503 & 0.543 & 0.075 & 0.072 & 94.9 \\
Misspecified, ML
& 600 & 0.533 & 0.505 & 0.537 & 0.049 & 0.049 & 95.2 \\
\bottomrule
\end{tabular}
\caption{\vspace{0.1em} Results for the one-step estimator $\hat\psi_t^{os}$ and plug-in estimator $\hat\psi_t^{pi}$. 
The first two rows correspond to data generated under the model used for the working models, whereas the remaining rows corresponds to data generated under a different model while keeping the same working models. Mean denotes the empirical mean across simulations, SD is the empirical standard deviation of the one-step estimator, ESE is the average estimated standard error for the one-step estimator, and Cov95 is the empirical coverage probability (in percent) of nominal 95\% confidence intervals based on the one-step estimator.}
\label{Tab1}
\end{table}

\subsection{Application to the Stanford Heart Transplant data}
\label{STH}
The data consist of patients with end-stage heart disease who were placed on a heart transplant waiting list. For each patient, follow-up begins at entry onto the waiting list and continues until death or censoring. During follow-up, some patients receive a heart transplant, while others die before transplantation or are censored.
These data have previously been analyzed using heart transplantation as a time-dependent treatment to correct earlier analyses where heart transplantation was taken as a baseline variable resulting in immortal-time bias. However, the interpretation of Cox-regression with heart transplantation as a time-dependent treatment is causally subtle as already mentioned.
Instead we apply the proposed methodology focussing on two treatment strategies: Treat from start and never treat. 
Although immediate transplantation at study entry is not practically feasible, the intervention may still be interpreted as a hypothetical strategy shifting transplantation as early as possible within the considered time window.
We took $b=30$ days meaning that the heart transplant should take place within the first 30 days after inclusion to the waiting list. It is seen from Figure \ref{fig:Data_STH0} in the Appendix that there is no significant difference between the two treatment strategies.
However, in \cite{kalbfleisch2002statistical}, it is suggested  that transplantation might interact with year of acceptance into the program as it is suggested that the overall quality of patients being admitted to the study  was improving with time. As an illustration, we investigate this using an indicator of whether a given patient is admitted 
after 1970-12-17 (half way through the period of admissions), which we  refer to as early/late. Figure \ref{fig:Data_STH1} shows the estimated risk of dying for the patients under the two treatment strategies: transplant from start or never transplant, and stratified by early/late admission into the program. The estimates indicate that the strategy of treating from the start versus never treat might be most beneficial for those who are admitted early into the program. However, as seen from Figure \ref{fig:Data_STH2}, this is not significant.

\begin{figure}
\centering
\includegraphics[width=10cm,height=8cm]{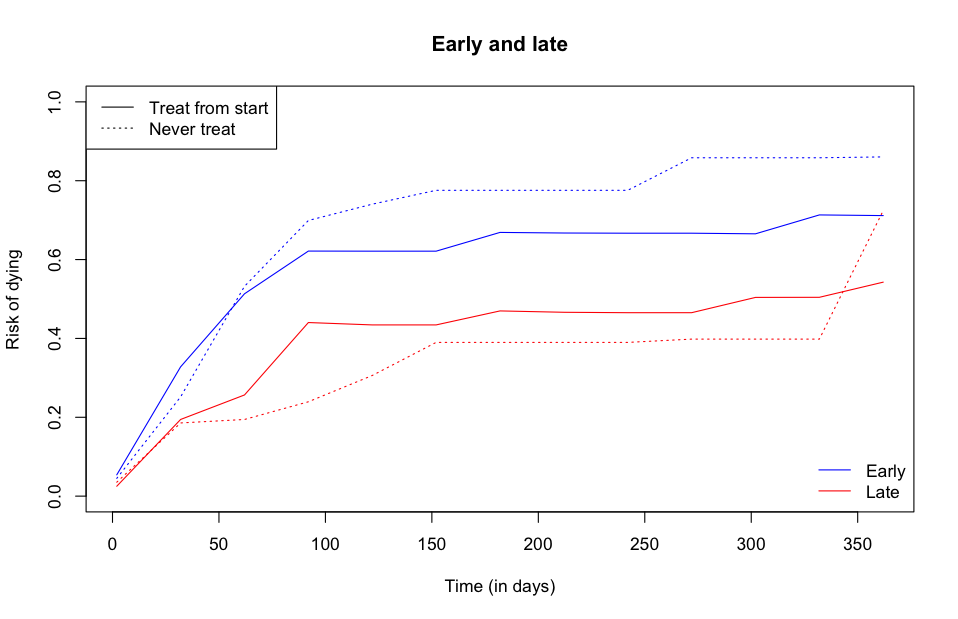}
\caption{Stanford Heart Transplant Data. Risk of dying as function of time under the two treatment strategies: "Transplant from start" and "Never Transplant" and stratified by early/late admission to the program.}
\label{fig:Data_STH1}
\end{figure}

\begin{figure}
\centering
\includegraphics[width=12cm,height=8cm]{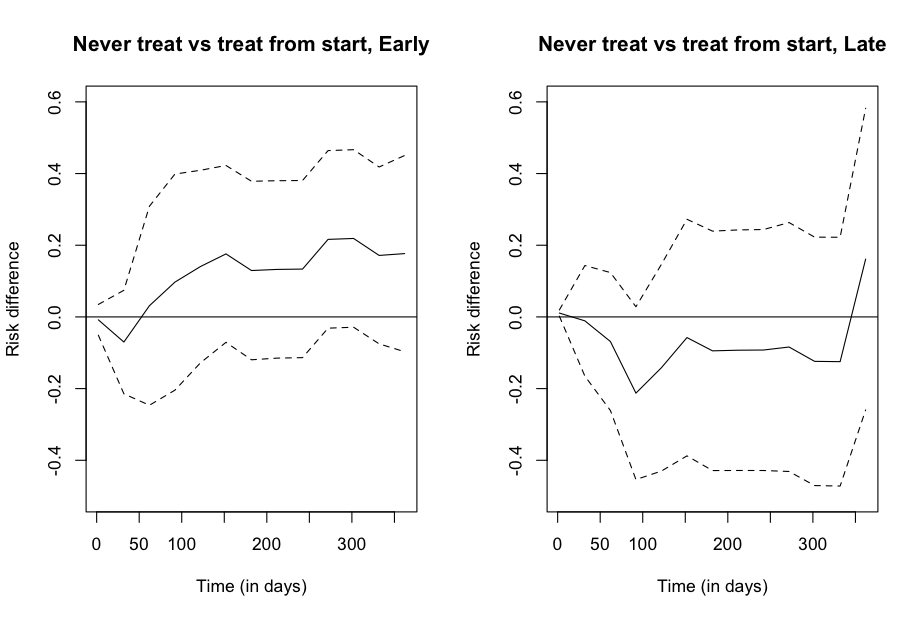}
\caption{Stanford Heart Transplant Data. Contrasts for the two treatment strategies and stratified by early/late admission to the program.}
\label{fig:Data_STH2}
\end{figure}


\subsection{Application to data on the effect of the timing of intrauterine insemination for couples with unexplained
subfertility
}

We also applied our method to data presented in \cite{prosepe2025causal} on time to 
achieving pregnancy  within 1.5 years after completion of the diagnostic workup under different strategies for delaying treatment initiation. The population consisted of couples with unexplained subfertility, defined as having attempted natural conception for more than one year without success despite normal findings on standard fertility assessments.
The question of whether to postpone treatment is clinically relevant, as assisted reproductive interventions may increase the likelihood of pregnancy, but are also associated with potential drawbacks. In particular, intrauterine insemination (IUI) may involve hormonal stimulation with possible side effects, as well as financial costs and psychological burden for the couple.
 More details of the study can be found in \cite{prosepe2025causal} and references therein. We use the same data consisting of 1896 couples, each followed for up to four years. Among these, 569 couples conceived without treatment, while 863 initiated IUI. Of those receiving IUI, 163 subsequently achieved pregnancy. As an illustration we consider three treatment strategies: Treat from start; Treat after 6 months and Never treat.
The analysis incorporated baseline characteristics collected during the diagnostic evaluation that could plausibly affect both the probability of future pregnancy and the timing of treatment initiation. The covariates considered were female age (centered), duration of subfertility,  gynecologist referral (yes/no), and the proportion of progressively motile sperm (centered). We took $b$ as two months. Figure \ref{fig:Data_Prosepe}, top row,
displays the probability of pregnancy under the three treatment strategies: never treat; treat from start and treat after 6 months. The figure also shows 95\% confidence intervals. The bottom row gives the three possible contrasts as indicated along with 95\% confidence intervals. It is seen that treatment from start is superior to never treat while no significant difference is seen between the strategies treat from start and treat after 6 months. Also, strategy treat after 6 months is superior to never treat.

\begin{figure}
\centering
\includegraphics[width=14.5cm,height=10cm]{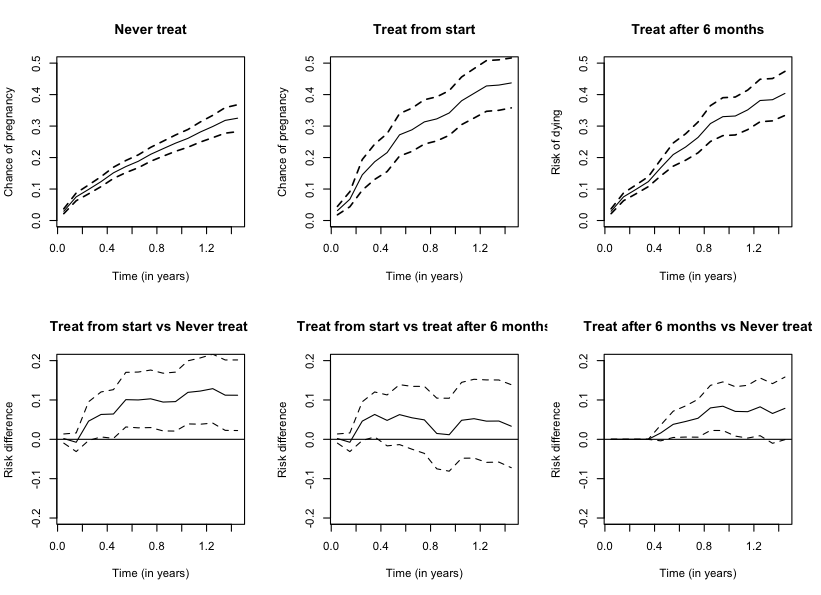}
\caption{IUI Data. Top panel from left to right: Chance of getting pregnant under the three different treatment strategies along with 95\% confidence intervals.
Bottom panel from left to right: The three possible contrasts as indicated along with 95\% confidence intervals.}
\label{fig:Data_Prosepe}
\end{figure}

\section{Concluding remarks}

In this paper, we have used the illness-death framework to deal with the situation where treatment status can change once over the considered time period. More complex treatment dynamics  may be formulated using more general 
multistate frameworks and a similar, but obviously more involved, method can be developed.
Similarly, we can also deal with a competing risk situation where the endpoint could for instance be time to relapse with death being a competing event. A similar development can be phrased but will again require a more involved multistate framework.

One can shift the focus from the  estimand  $\tilde P^{o}(T\leq t)$ to another summary measure such as
$$
\int_0^{\tau}\tilde P^{o}(T> t)dt=\tilde E^{o}(T\wedge\tau)
$$
corresponding to the restricted mean survival time under the considered intervention. Since this latter estimand is a simple function of the estimand considered in the paper, the corresponding theoretical development follows readily.

Although the proposed framework is compatible with flexible nuisance estimation in principle, fully data-adaptive estimation of the post-treatment transition component remains challenging in the present delayed-entry setting. The development of practically stable machine-learning procedures for this component constitutes an interesting topic for future work.

\bibliographystyle{biometrika}
\bibliography{references}

\clearpage

\pagenumbering{arabic}  
\pagestyle{headings}
$\mbox{}$

\section*{Appendix}

We have the following lemma:

\begin{lem}
\label{basic_eif}
The observed data efficient influence functions corresponding to $A_{01}(t)$,  $A_{02}(t)$ and $A_{12}(t|u)$ are, respectively,
\begin{align*}
    D_{A_{01}(t)}^*(O,P)&=\int_0^t\frac{dM_{01}^{o}(u)}{S_0(u)K(u)},\quad
   D_{A_{02}(t)}^*(O,P) =\int_0^t\frac{dM_{02}^{o}(u)}{S_0(u)K(u)},\\
    D_{A_{12}(t|u)}^*(O,P)&=\frac{I(U\wedge T \wedge C=u)\delta_1}{P(U\wedge T \wedge C=u,\delta_1=1)}\int_u^t\frac{dM_{12}^{o}(r|u)}{S_1(r|u)K(r|u)}.\\
\end{align*}
Also,
 $
 P(U\wedge T \wedge C=u,\delta_1=1)=S_0(u)K(u)\alpha_{01}(u).
 $
\end{lem}

\noindent{\it Proof of Theorem \ref{thm:eif}}\quad
\vspace{-0.3cm}

Define
\begin{align*}
h_{\eta}^{(1)}(u,t)c^*=&\left\{1-S_1(t|u)
\right\} S_0(u)\\
h_{\mu}^{(1)}(u,t)c^*=&e^{-A_{01}(u_+^*)}
\left\{e^{-A_{02}(u_-^*)}-e^{-A_{02}(t\wedge u_+^*)}\right\}\\
&-\int_{u\vee u_-^*}^{t\wedge u_+^*} S_0(v)dA_{02}(v)
\\
h_{\mu}^{(2)}(u,t)c^*=&I\bigl( u>u_-^*\bigr )H_{01}^*(u)e^{-A_{02}(u)}
-\int_{u\vee u_-^*}^{t\wedge u_+^*} H_{01}^*(v)e^{-A_{02}(v)}dA_{02}(v)
\end{align*}
The claimed result follows from Lemma \ref{basic_eif}, the tips and tricks of \cite{kennedy2022semiparametric} and some algebra.
\hfill $\square$

\noindent{\it Proof of Theorem \ref{thm:remainder}}\quad
\vspace{-0.3cm}

The remainder term is given by 
\begin{equation}
\label{rem-gen}
    R_n=\psi_t(P_n)-\psi_t(P)+E\{D_{\psi_t}^*(O,P_n)\}.
\end{equation}
The key to calculate $E\{D_{\psi_t}^*(O,P_n)\}$ is that, for $j=1,2$,
$$
E\left\{\frac{dM_{0j}^n(u)}{S_0^n(u)K_n(u)}\right\}=-\frac{S_0(u)K(u)}{S_0^n(u)K_n(u)}d\{A^n_{0j}-A_{0j}\}(u)
$$
and
$$
E\left\{\frac{dM_{12}^n(r|u)}{S_1^n(r|u)K_n(r|u)}\right\}=-\frac{S_1(r|u)K(r|u)}{S_1^n(r|u)K_n(r|u)}d\{A^n_{12}-A_{12}\}(r|u)
$$
and if we first assume that $K$ is known, so $K_n=K$, then, after some algebra, it is seen that $R_n$ consists of cross-product terms of the kind
\begin{align*}
&\{S_1^n(t|u)-S_1(t|u)\}\{S_{02}^n(u)-S_{02}(u)\}, \quad \{S_1^n(t|u)-S_1(t|u)\}\{S_{01}^n(u)dA_{01}^n(u)-S_{01}(u)dA_{01}(u)\}\\
&\{S_{02}^n(u)-S_{02}(u)\}\{S_{01}^n(u)dA_{01}^n(u)-S_{01}(u)dA_{01}(u)\}, \quad \{S_{02}^n(u)-S_{02}(u)\}\{dA_{01}^n(u)-dA_{01}(u)\}\\
&\{S_{01}^n(u)-S_{01}(u)\}\{dA_{02}^n(u)-dA_{02}(u)\}
\end{align*}
that are all of low order when Nelson-Aalen estimators are used to estimate $A_{0j}$, $j=1,2.$ Since $K$ is not known we re-write
$$
K/K_n=1-\frac{K_n-K}{K_n}
$$
and from the "1" in the latter display we get the exact same terms as listed above where we took $K$ as known. In addition, we get terms that also involve  $K_n-K$. Specifically, these latter terms are all of the kind
$$
\frac{K_n(u)-K(u)}{K_n(u)}d\{ A_{0j}^n(u)- A_{0j}(u) \},\; j=1,2,
$$
and 
$$
\frac{K_n(r|u)-K(r|u)}{K_n(r|u)}d\{ A_{12}^n(r|u)- A_{12}(r|u) \}.
$$
It is only the latter term that matters because we use Nelson-Aalen estimators to estimate $A_{0j}$, $j=1,2$.
The specific expression of the latter term is 
$$
\int J_t^*(u)\left\{\int_u^t\frac{K_n(r|u)-K(r|u)}{K_n(r|u)}\frac{S_1(r|u)}{S^n_1(r|u)}d\{ A_{12}^n(r|u)- A_{12}(r|u) \}\right\}
S_1^n(t|u)S_0^n(u)dA^n_{01}(u)
$$

\hfill $\square$

\noindent{\it Proof of Theorem \ref{thm:remainder_w_L}}\quad
\vspace{-0.3cm}

It follows directly from the proof of \ref{thm:remainder} that the remainder term consists of 
cross-product terms of the kind
\begin{align*}
&\{S_1^n(t|u,l)-S_1(t|u,l)\}\{S_{02}^n(u|l)-S_{02}(u|l)\}, \; \{S_1^n(t|u,l)-S_1(t|u,l)\}\{S_{01}^n(u|l)dA_{01}^n(u|l)-S_{01}(u|l)dA_{01}(u|l)\}\\
&\{S_{02}^n(u|l)-S_{02}(u|l)\}\{S_{01}^n(u|l)dA_{01}^n(u|l)-S_{01}(u|l)dA_{01}(u|l)\}, \; \{S_{02}^n(u|l)-S_{02}(u|l)\}\{dA_{01}^n(u|l)-dA_{01}(u|l)\}\\
&\{S_{01}^n(u|l)-S_{01}(u|l)\}\{dA_{02}^n(u|l)-dA_{02}(u|l)\}
\end{align*}
and of 
$$
\frac{K_n(u|l)-K(u|l)}{K_n(u|l)}d\{ A_{0j}^n(u|l)- A_{0j}(u|l) \},\; j=1,2,
$$
and 
$$
\frac{K_n(r|u,l)-K(r|u,l)}{K_n(r|u,l)}d\{ A_{12}^n(r|u,l)- A_{12}(r|u,l) \}.
$$
For example, the specific expression of the latter term is 
\begin{align*}
E\int J_t^*(u)&\left\{\int_u^t\frac{K_n(r|u,L)-K(r|u,L)}{K_n(r|u,L)}\frac{S_1(r|u,L)}{S^n_1(r|u,L)}d\{ A_{12}^n(r|u,L)- A_{12}(r|u,L) \}\right\}\\
&\times S_1^n(t|u,L)S_0^n(u|L)dA^n_{01}(u|L)
\end{align*}

\nothere{
The remainder term is given by 
$$
R_n=\psi_t(P_n)-\psi_t(P)+E\{D_{\psi_t}^*(O,P_n)\}.
$$
If the censoring survival function was known then the one-step estimator is consistent and has IF equal to the EIF $D_{\psi_t}^*(O,P)$.
Since, in reality,  the censoring survival function is unknown we need to replace $K$ with some estimate $K_n$.
However, because
$$
K/K_n=1-\frac{K_n-K}{K_n}\approx 1+(\Lambda_C^n-\Lambda_C),
$$
we get from the "1" the results corresponding to the censoring  survival function being known and the last term in the latter display appears then in all the (estimated) martingale terms of $D_{\psi_t}^*(O,P_n)$ leading to the terms given in Theorem \ref{thm:remainder}. We exemplify this with one concrete calculation. Consider the (estimated) martingale
\begin{align*}
  &E\left [\int J^*_t(u)h_{\eta}^{(1),n}(u,t|L)\frac{dM_{01}^{o,n}(u|L)}{S_0^n(u|L)K_n(u|L)}\right ]\\
=&E\left [\int J^*_t(u)h_{\eta}^{(1),n}(u,t|L)\frac{S_0(u|L)K(u|L)}{S_0^n(u|L)K_n(u|L)}d\{A_{01}^n-A_{01}\}(u|L)\right ]\\
=&E\left [\int J^*_t(u)h_{\eta}^{(1),n}(u,t|L)\frac{S_0(u|L)}{S_0^n(u|L)}d\{A_{01}^n-A_{01}\}(u|L)\right ]+\\
&E\left [\int J^*_t(u)h_{\eta}^{(1),n}(u,t|L)\frac{S_0(u|L)\{K_n-K\}(u|L)}{S_0^n(u|L)K_n(u|L)}d\{A_{01}^n-A_{01}\}(u|L)\right ],
\end{align*}
where first term in the last sum of the latter display cancels with a corresponding term in $\psi_t(P_n)-\psi_t(P)$. All martingale terms are handled in the same way and we are left with the remainder terms as stated in Theorem \ref{thm:remainder}.
}
\hfill $\square$

\begin{figure}
\centering
\includegraphics[width=14cm,height=10cm]{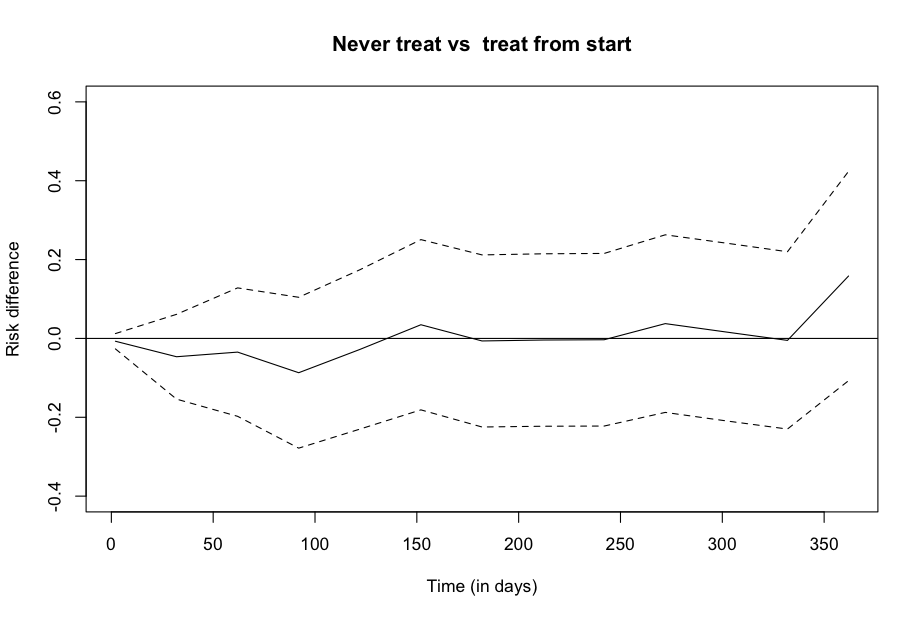}
\caption{Stanford Heart Transplant Data. Risk difference along with 95\%-confidence intervals under the two treatment strategies: "Transplant from start" and "Never Transplant".}
\label{fig:Data_STH0}
\end{figure}

\end{document}